\documentclass[journal,10pt]{IEEEtran}

\usepackage{amssymb}
\usepackage{amsmath}
\usepackage{cite}
\usepackage{url}
\usepackage{xcolor}
\usepackage{cite,graphicx,amsmath,amssymb}
\usepackage{subfigure}
\usepackage{fancyhdr}
\usepackage{mdwmath}
\usepackage{mdwtab}
\usepackage[caption=false,font=normalsize,labelfont=sf,textfont=sf]{subfig}
\usepackage{amsmath}
\usepackage{bbm}
\usepackage{epstopdf}

\newtheorem{remark}{Remark}
\newtheorem{theorem}{Theorem}

\newtheorem{lemma}{Lemma}

\newtheorem{corollary}{Corollary}

\newenvironment{proof}{\indent \indent \it Proof:}{\hfill $\blacksquare$\par}

\makeatletter
\def\ScaleIfNeeded{%
\ifdim\Gin@nat@width>\linewidth \linewidth \else \Gin@nat@width
\fi } \makeatother

\begin{document}

\title{Ergodic Rate Analysis of STAR-RIS Aided \\NOMA Systems}

\author{
Boqun~Zhao,~\IEEEmembership{Graduate Student Member,~IEEE},
Chao~Zhang,~\IEEEmembership{Graduate Student Member,~IEEE},
Wenqiang~Yi,~\IEEEmembership{Member,~IEEE} and
Yuanwei~Liu,~\IEEEmembership{Senior Member,~IEEE}

\thanks{B. Zhao is with Department of Electrical and Electronic Engineering, Imperial College London, London, UK (email:burch.zhao21@imperial.ac.uk).}
\thanks{C. Zhang, W. Yi and Y. Liu are with the School of Electronic Engineering and Computer Science, Queen Mary University of London, London, UK (email:\{chao.zhang, w.yi, yuanwei.liu\}@qmul.ac.uk).}
}

\maketitle
\begin{abstract}
  This letter analyzes the ergodic rates of a simultaneously transmitting and reflecting reconfigurable intelligent surface (STAR-RIS) aided non-orthogonal multiple access (NOMA) system, where the direct links from the base station to cell-edge users are non-line-of-sight due to obstacles, and STAR-RIS is used to provide line-of-sight links to these cell-edge users. By fitting the distribution of the composite channel power gain to a gamma distribution, we derive the closed-form expressions of ergodic rates and high signal-to-noise ratio (SNR) slopes for cell-edge users. Numerical results reveal that 1) the ergodic rates increase with the number of STAR-RIS elements, and the high SNR slopes are fixed as constants; 2) STAR-RIS aided NOMA systems achieve higher ergodic rates than conventional RIS aided NOMA systems.
\end{abstract}

\begin{IEEEkeywords}
STAR-RIS, NOMA, ergodic rate, high SNR slope 
\end{IEEEkeywords}

\section{Introduction}
Due to the capability of supporting high throughput and massive connectivity, non-orthogonal multiple access (NOMA) has been recognized as a promising multiple access technique for next generation wireless networks \cite{NOMA}. However, cell-edge NOMA users experience low quality of service (QoS) due to additional intra-NOMA-cluster interference. Additionally, the transmission links from base stations (BSs) to cell-edge users prone to be non-line-of-sight (NLoS) due to obstacles, which further impairs their QoS. To solve this issue, an emerging technique, named reconfigurable intelligent surfaces (RISs), is employed to enhance the performance of cell-edge users via controlling their surrounding wireless propagation environment \cite{NOMA-ris}. 
  
However, conventional RISs with opaque substrates may create new blind spots behind the RISs. Recently proposed simultaneously transmitting and reflecting reconfigurable intelligent surfaces (STAR-RISs) have the ability to overcome this limitation. In contrast to conventional RISs, STAR-RISs can reflect and refract signals to users on both sides of their surfaces, thereby enlarging the serving area from $180^{\circ }$ to $360^{\circ }$ coverage \cite{refraction}. Hence, applying STAR-RISs to NOMA systems is able to achieve full-space coverage without any blind spots, which brings STAR-RIS aided NOMA a high research value. 
    \begin{figure}[!htb]
    \centering
    \includegraphics[width= 3.4in]{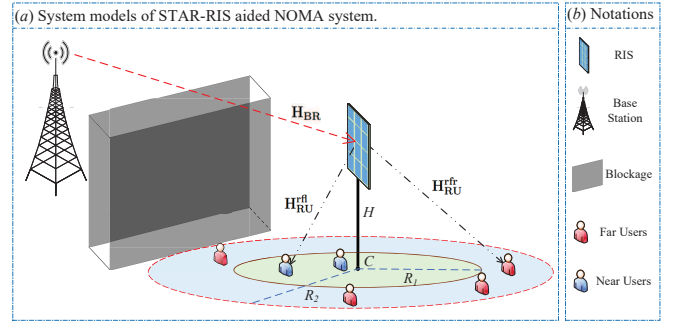}
    \caption{System models of the STAR-RIS-aided NOMA system: a) System channels and user deployment; and b) Notations.}
    \label{Fig1}
    \end{figure}
However, though there are some existing works regarding the performance of STAR-RIS aided NOMA systems \cite{ios,Ergodic,gamma,R4}, the mathematical analyses of ergodic rate are only performed in \cite{Ergodic} and \cite{gamma}, and neither derives the closed-form expressions of the ergodic rates for both NOMA pairs. In addition, the models proposed in \cite{Ergodic} and \cite{gamma} are both based on Rician fading channels suitable for an ideal communication environment, and the line-of-sight (LoS) links from BS to users are available in the models, which is impractical for the cell-edge area.
  
Motivated by the aforementioned research gap, we investigate the closed-form ergodic rate expressions of a STAR-RIS aided NOMA system without LoS links from the BS to users. To characterize the statistical property of actual rapid fading channels, a more practical channel, e.g., Nakagami-$m$ fading, is considered. The channel model for the system is mimicked with a gamma distribution, which is friendly for multi-cell performance analysis \cite{gamma}. The main contributions are summarized as follows: 1) The closed-form expressions of ergodic rates for cell-edge users are derived; 2) The approximated ergodic rates in high signal-to-noise ratio (SNR) region are calculated to investigate the high SNR slopes; 3) Numerical results validate that the ergodic rates of STAR-RIS aided NOMA systems outperform conventional RIS aided NOMA systems.

\section{System model}
By considering a cell-edge area where the transmission links from BS to users are NLoS, we investigate a STAR-RIS aided NOMA system, including a BS, randomly deployed users, and a STAR-RIS. As shown in Fig. \ref{Fig1}, we denote the channel as ${\mathbf{H}}_{\mathbf{BR}}$ for the link from the BS to the STAR-RIS, ${\mathbf{H}}_{{\mathbf{RU}}}^{\mathbf{rfl}}$ and ${\mathbf{H}}_{{\mathbf{RU}}}^{\mathbf{rfr}}$ for the link from the STAR-RIS to the reflecting user and the refracting user, respectively. Detailed system model designs are presented in the following.

\subsection{Deployment}
As shown in Fig. \ref{Fig1}, the BS is fixed and the direct links from the BS to users are blocked. We assume the users are deployed in a circle area with the center point denoted as $C$. The STAR-RIS is above the center $C$ with a height of $H$ meters. Since the distance from the BS to the users is much larger than the upper bound of the near field, the far field is assumed. More specifically, to evaluate the performance for the users with different distances to the STAR-RIS, the users are split into two portions, the near user deployed in the inner circle with a radius $R_1$ meters and the far users deployed in the outer ring with radii $R_2$ and $R_1$ meters. We exploit homogeneous Poisson point processes to model the locations of the users \cite{ppp}. Hence, the near and far users are uniformly distributed within their areas and the probability density functions (PDF) of the distances from a user to the center are derived as ${f_{{d_{near}}}}\left( x \right) = \frac{\partial }{{\partial x}}\frac{{\pi {x^2}}}{{\pi R_1^2}} = \frac{{2x}}{{R_1^2}}$ and ${f_{{d_{far}}}}\left( x \right) = \frac{\partial }{{\partial x}}\frac{{\pi \left( {{x^2} - R_1^2} \right)}}{{\pi \left( {R_2^2 - R_1^2} \right)}} = \frac{{2x}}{{R_2^2 - R_1^2}}$.

\subsection{Protocol and NOMA Designs}
\subsubsection{The Energy Splitting (ES) Protocol}
By utilizing the ES protocol, all the STAR-RIS elements can split the total radiation energy into the refraction and reflection modes. Since the passive beamforming consumes little energy, we assume an ideal scenario where the total energy is exploited. Hence, for the $n^{th} $ element, this protocol is mathematically presented as $\beta _n^{rfl} + \beta _n^{rfr} = 1$, where $\beta _n^{rfl}$ is for the reflecting component and $\beta _n^{rfr}$ is for the refracting component \cite{ios}. To avoid the ES protocol changing the arranged orders of the successive interference cancellation (SIC), we assume $\beta _n^{rfr}=\beta _n^{rfl}=0.5$ for the following analysis.

\subsubsection{NOMA Designs}
To maximise the gain for each user and consider unicast transmission, we assume that a near user and a far user on the different sides of the STAR-RIS are paired in a NOMA pair. When considering the average performance, the near users have better channel quality than far users. The BS will allocate more energy for the far users than the near users. Thus, the near users will utilize the SIC process, while the far users will directly decode their signals. We denote the power allocation coefficients as $a_{near}<a_{far}$ and $a_{near}+a_{far}=1$.

\subsection{Channel Models}
\subsubsection{STAR-RIS Equivalent Channels}
For the better performance, the optimal equal phase configuration for the STAR-RIS is considered. The reflecting and refracting channels from the BS to the user via the STAR-RIS with $N$ elements are presented as 
\begin{align}
 \left| {{h_{rfl}}} \right| = \left| {{\mathbf{H}}{{_{{\mathbf{RU}}}^{\mathbf{rfl}}}^T}{{\mathbf{\Theta }}_{{\mathbf{rfl}}}}{{\mathbf{H}}_{{\mathbf{BR}}}}} \right|,  \left| {{h_{rfr}}} \right| = \left| {{\mathbf{H}}{{_{{\mathbf{RU}}}^{{\mathbf{rfr}}}}^T}{{\mathbf{\Theta }}_{{\mathbf{rfr}}}}{{\mathbf{H}}_{{\mathbf{BR}}}}} \right|,
\end{align}
where ${{\bf{\Theta }}_{{\bf{rfl}}}} = diag\left[ {\sqrt {\beta _1^{rfl}} {e^{j\phi _1^{rfl}}}, \cdots ,\sqrt {\beta _N^{rfl}} {e^{j\phi _N^{rfl}}}} \right]$ and ${{\bf{\Theta }}_{{\bf{rfr}}}} = diag\left[ {\sqrt {\beta _1^{rfr}} {e^{j\phi _1^{rfr}}},\cdots ,\sqrt {\beta _N^{rfr}} {e^{j\phi _N^{rfr}}}} \right]$ are the diagonal phase-shift matrices of the STAR-RIS elements for reflecting and refracting channels, ${\bf{H}}_{{\bf{RU}}}^{{\bf{rfl}}} = {\left[ {h_{RU,1}^{rfl} , \cdots ,h_{RU,N}^{rfl}} \right]^T}$ and ${\bf{H}}_{{\bf{RU}}}^{{\bf{rfr}}} = {\left[ {h_{RU,1}^{rfr}, \cdots ,h_{RU,N}^{rfr}} \right]^T}$ are the channels between the STAR-RIS and the users, and ${{\bf{H}}_{{\bf{BR}}}} = {\left[ {{h_{BR,1}} , \cdots ,{h_{BR,N}}} \right]^T}$ is the channel between the BS and the STAR-RIS.

\subsubsection{Small-Scale Fading Model}
We assume the channels for each STAR-RIS element are independent and identical distributed (i. i. d.) Nakagami-$m$ fading channels\footnote{The results for the spatially correlated channels are evaluated by simulation in Section \uppercase\expandafter{\romannumeral4}. The independent channels can be seen as a special case of the correlated channels when the spatial correlation matrix is identity or $N \to \infty $ \cite{corr}. Therefore, the performance under the i.i.d. channels can be used as a benchmark in the future studies of the spatially correlated channels.}, and multiple elements introduce the composite channel. When $N$ is large, the distribution of the composite channel power gain ${\left| {{h_i}} \right|^2}$, with $i \in \{rfl, rfr\}$, can be fitted as a gamma distribution $\Gamma \left( {k,\theta } \right)$ with the shape parameter and the scale parameter expressed as follows \cite{gamma}:
\begin{equation}
\begin{cases}
k = \frac{{\Gamma {{\left( {m + \frac{1}{2}} \right)}^2}}}{{4\left( {m\Gamma {{\left( m \right)}^2} - \Gamma {{\left( {m + \frac{1}{2}} \right)}^2}} \right)}}N \\
\theta = 4\Omega N - \frac{{4\Omega }}{m}{\left( {\frac{{\Gamma \left( {m + \frac{1}{2}} \right)}}{{\Gamma \left( m \right)}}} \right)^2}N
\end{cases},
\label{equ2}
\end{equation}
where $m$ and $\Omega$ are the shape parameter and the spread parameter of the Nakagami-$m$ distribution.

\subsubsection{Signal-to-Interference-and-Noise Ratio (SINR)}
For this scenario, the SINR or SNR expressions of the near and far users are presented as
\begin{align}
\label{equ3}
{\gamma ^{SIC}_{{near}}} &= \frac{{{P_t}{a_{far}}(d_{near}^2+H^2)^{-\frac{\alpha}{2}}d_{BR}^{-\alpha}{{\left| {h_k} \right|}^2}}}{{{P_t}{a_{near}}(d_{near}^2+H^2)^{-\frac{\alpha}{2}}d_{BR}^{-\alpha}{{\left| {h_k} \right|}^2} + {\sigma ^2}}},\\ 
\label{equ4}
{\gamma _{near}} &= \frac{{P_t}{a_{near}}(d_{near}^2+H^2)^{-\frac{\alpha}{2}}d_{BR}^{-\alpha}{{\left| {h_k} \right|}^2}}{{ {\sigma ^2}}},\\ 
{\gamma _{far}} &= \frac{{{P_t}{a_{far}}(d_{far}^2+H^2)^{-\frac{\alpha}{2}}d_{BR}^{-\alpha}{{\left| {h_j} \right|}^2}}}{{{P_t}{a_{near}}(d_{far}^2+H^2)^{-\frac{\alpha}{2}}d_{BR}^{-\alpha}{{\left| {h_j} \right|}^2} + {\sigma ^2}}},
\end{align}
where $P_{t}$ is the transmit power of users, $\sigma^{2}$ is the variance of additive white Gaussian noise (AWGN) and $\alpha$ is the path loss exponent. When the near user is the reflecting user, we have $k = rfl$ and $j = rfr$. When the near user is the refracting user, we have $k = rfr$ and $j = rfl$.

\section{Ergodic rate analysis}
The ergodic rate that indicates the averaged achievable rate is an essential metric for performance evaluation. In this section, we investigate the ergodic rates for the near and far users of the proposed STAR-RIS aided NOMA system, which are defined as
\begin{align}
 &{R_{near}} = {\mathbb{E}}\left[ {\log \left( {1 + {\gamma _{near}}} \right)} \right]u\left( {\gamma _{near}^{SIC} - \gamma _{th}^{SIC}} \right), \\
 &{R_{far}} = {\mathbb{E}}\left[ {\log \left( {1 + {\gamma _{far}} } \right)} \right],
\end{align}
where $\mathbb{E}\left[  \cdot  \right]$ denotes the expectation calculation, $u\left(t\right)$ is the unit step function defined as $u\left( t \right) = \left\{ {\begin{array}{*{20}{c}}
  \!\!{0,~~t < 0} \\ 
  \!\!{1,~~t \geqslant 0} 
\end{array}} \right.$, and $\gamma _{th}^{SIC}$ is the SIC threshold.

More specifically, we assume the number of STAR-RIS elements is large enough, and with the aid of (2), the distribution of the composite channel power gain is fitted as a gamma distribution $\Gamma \left( {k,\theta } \right)$. Based on this approximation, the expressions of the ergodic rates and high SNR slopes are derived as follows.

\subsection{Ergodic rate of near user}
\begin{theorem}\label{theorem1}
With the aid of the ES protocol, the closed-form expression of the ergodic rate for the near users is derived as
\setlength{\abovedisplayskip}{2pt}
\setlength{\belowdisplayskip}{2pt}
\begin{align}\label{nearER}
  {R_{near}} = {I_1} + {I_2} + {I_3}+ {I_4},   
\end{align}
where $I_1$, $I_2$, $I_3$ and $I_4$ denote as
\begin{align}
{I_1} &= \frac{{\ln \left( {\Xi  + 1} \right)}}{{{R_1}\ln 2}}\sum\limits_{i = 1}^M {\omega _i}\sqrt {1 - {\varepsilon _i}^2} {\Omega _1}\left( {{\varepsilon _i}} \right) \notag\\
&~~~\times \left( {1 - \frac{{\gamma \left( {k,\Xi \Psi \left( {{\Omega _1}\left( {{\varepsilon _i}} \right)} \right)} \right)}}{{\Gamma \left( k \right)}}} \right),\\
{I_2} &=  - \frac{1}{{{R_1}\ln 2}}\sum\limits_{i = 1}^M {\omega _i}\sqrt {1 - {\varepsilon _i}^2} {\Omega _1}\left( {{\varepsilon _i}} \right)\exp \left( {\Psi \left( {{\Omega _1}\left( {{\varepsilon _i}} \right)} \right)} \right)\notag\\
&~~~\times {\text{Ei}}\left( { - \left( {\Xi  + 1} \right)\Psi \left( {{\Omega _1}\left( {{\varepsilon _i}} \right)} \right)} \right),\\
{I_3} &= \frac{1}{{{R_1}\ln 2}}\sum\limits_{i = 1}^M \sum\limits_{n = 1}^{k - 1} {\omega _i}\sqrt {1 - {\varepsilon _i}^2} {\Omega _1}\left( {{\varepsilon _i}} \right) \frac{{{\Psi ^n}\left( {{\Omega _1}\left( {{\varepsilon _i}} \right)} \right)}}{{n!}}\notag\\
&~~~\times {{\left( { - 1} \right)}^{n + 1}}\exp \left( {\Psi \left( {{\Omega _1}\left( {{\varepsilon _i}} \right)} \right)} \right)\notag\\
&~~~\times{\text{Ei}}\left( { - \left( {\Xi  + 1} \right)\Psi \left( {{\Omega _1}\left( {{\varepsilon _i}} \right)} \right)} \right),
\end{align}
\begin{align}
{I_4} &= \frac{1}{{{R_1}\ln 2}}\sum\limits_{i = 1}^M \sum\limits_{n = 1}^{k - 1} \sum\limits_{j = 1}^n {\omega _i}\sqrt {1 - {\varepsilon _i}^2} {\Omega _1}\left( {{\varepsilon _i}} \right)\frac{{{\Psi ^n}\left( {{\Omega _1}\left( {{\varepsilon _i}} \right)} \right)}}{{n!}}\notag\\
&~~~\times\exp \left( {\Psi \left( {{\Omega _1}\left( {{\varepsilon _i}} \right)} \right)} \right) C_n^j{{\left( { - 1} \right)}^{n - j}}{{\left( {\Psi \left( {{\Omega _1}\left( {{\varepsilon _i}} \right)} \right)} \right)}^{ - j}} \notag\\
&~~~\times \Gamma \left( {j,\left( {\Xi  + 1} \right)\Psi \left( {{\Omega _1}\left( {{\varepsilon _i}} \right)} \right)} \right)  ,
\end{align}
where ${\Omega _1}\left( x \right) = \frac{{{R_1}}}{2}\left( {x + 1} \right)$, ${\varepsilon _i} = \cos \left( {\frac{{2i - 1}}{{2M}}\pi } \right)$, ${\omega _i} = \frac{\pi }{M}$, $M$ is the coefficients in Chebyshev-Gauss quadrature, $\Xi  = \frac{{{a_{near}}\gamma _{th}^{SIC}}}{{\left( {{a_{far}} - \gamma _{th}^{SIC}{a_{near}}} \right)}}$, $\Psi \left( y \right)=\frac{{{\sigma ^2}}}{{{a_{near}}{P_t}{{({y^2} + {H^2})}^{ - \frac{\alpha }{2}}}d_{BR}^{ - \alpha }\theta }}$, ${\text{Ei}}\left( x \right) = \int_{ - \infty }^x {\frac{{e^t}}{t}dt}$ is the exponential integral, $\Gamma \left( k \right) = \int_0^\infty  {{t^{k - 1}}{e^{ - t}}dt}$ is the gamma function, $\gamma \left( {k,x} \right) = \int\limits_0^x {{t^{k - 1}}{e^{ - t}}dt}$ is the lower incomplete gamma function and $\Gamma \left( {k,x} \right) = \int_x^\infty  {{t^{k - 1}}{e^{ - t}}dt}$ is the upper incomplete gamma function.
\end{theorem}
\vspace{0.1cm}
\begin{proof}
\emph{See Appendix~A.}
\end{proof}

\subsection{Ergodic rate of far user}
\begin{theorem}\label{theorem2}
The closed-form expressions of the ergodic rate for the far users is derived as
\begin{align}
{R_{far}} &= \frac{{{a_{far}}}}{{{a_{near}}\left( {R_2^2 - R_1^2} \right)\alpha \ln 2}} \notag\\
&~~~\times \sum\limits_{i = 1}^M \sum\limits_{n = 0}^{k - 1} {\omega _i}\sqrt {1 - {\varepsilon _i}^2} \frac{{\Phi^{ - \frac{2}{\alpha }}\left( {{\Omega _2}\left( {{\varepsilon _i}} \right)} \right)}}{{n!\left( {1 + {\Omega _2}\left( {{\varepsilon _i}} \right)} \right)}}\notag\\
&~~~\times \left( \gamma \left( {n + \frac{2}{\alpha },{\Phi}\left( {{\Omega _2}\left( {{\varepsilon _i}} \right)} \right){{({R_2}^2 + {H^2})}^{\frac{\alpha }{2}}}} \right) 
\notag\right.
\\
\phantom{=\;\;}
&~~~\left.- \gamma \left( {n + \frac{2}{\alpha },{\Phi}\left( {{\Omega _2}\left( {{\varepsilon _i}} \right)} \right){{({R_1}^2 + {H^2})}^{\frac{\alpha }{2}}}} \right) \right)  ,
\end{align}
where ${\Omega _2}\left( x \right) = \frac{{{a_{far}}}}{{2{a_{near}}}}\left( {x + 1} \right)$ and ${\Phi}\left( x \right) = \frac{{{\sigma ^2}x}}{{{P_t}d_{BR}^{ - \alpha }\theta \left( {{a_{far}} - {a_{near}}x} \right)}}$.
\end{theorem}
\vspace{0.2cm}
\begin{proof}
\emph{See Appendix~B.}
\end{proof}

\subsection{Approximated analysis on ergodic rate}
In this subsection, we investigate the approximated ergodic rates in high SNR region ($\rho  = {{{P_t}} \mathord{\left/{\vphantom {{{P_t}} {{\sigma ^2} \to }}} \right.\kern-\nulldelimiterspace} {{\sigma ^2} \to }}\infty$) as well as the high SNR slopes.

\vspace{0.1cm}
\subsubsection{Approximated ergodic rates}
\begin{corollary}
In approximated analysis,  we assume the SNR $\rho  = {{{P_t}} \mathord{\left/{\vphantom {{{P_t}} {{\sigma ^2} \to }}} \right.\kern-\nulldelimiterspace} {{\sigma ^2} \to }}\infty$. The closed-form approximated ergodic rates for the near is expressed as
\setlength{\abovedisplayskip}{3pt}
\setlength{\belowdisplayskip}{3pt}
\begin{align}
R_{near}^\infty  = I_1^\infty  + I_2^\infty  + I_3^\infty  + I_4^\infty 
\end{align}
where $I_1^\infty$, $I_2^\infty$, $I_3^\infty$ and $I_4^\infty$ are defined as
\begin{align}
{I_1^\infty} &= \frac{{\ln \left( {\Xi  + 1} \right)}}{{{R_1}\ln 2}}\sum\limits_{i = 1}^M {\omega _i}\sqrt {1 - {\varepsilon _i}^2} {\Omega _1}\left( {{\varepsilon _i}} \right) \notag\\
&~~~\times {\left( {1 - \frac{{{\Xi ^k}\Psi^k\left( {{\Omega _1}\left( {{\varepsilon _i}} \right)} \right)}}{{k!}}} \right)} ,\\
{I_2^\infty} &=  - \frac{1}{{{R_1}\ln 2}}\sum\limits_{i = 1}^M {\omega _i}\sqrt {1 - {\varepsilon _i}^2} {\Omega _1}\left( {{\varepsilon _i}} \right)\left( {1 + \Psi \left( {{\Omega _1}\left( {{\varepsilon _i}} \right)} \right)} \right)\notag\\
&~~~\times \left( {{\text{ln}}\left( {\left( {\Xi  + 1} \right)\Psi \left( {{\Omega _1}\left( {{\varepsilon _i}} \right)} \right)} \right) + C} \right),\\
{I_3^\infty} &= \frac{1}{{{R_1}\ln 2}}\sum\limits_{i = 1}^M \sum\limits_{n = 1}^{k - 1} {\omega _i}\sqrt {1 - {\varepsilon _i}^2} {\Omega _1}\left( {{\varepsilon _i}} \right)\frac{{{\Psi ^n}\left( {{\Omega _1}\left( {{\varepsilon _i}} \right)} \right)}}{{n!}} \notag\\
&~~~\times {{\left( { - 1} \right)}^{n + 1}}{\left( {1 + \Psi \left( {{\Omega _1}\left( {{\varepsilon _i}} \right)} \right)} \right)} \notag\\
&~~~\times {\left( {{\text{ln}}\left( {\left( {\Xi  + 1} \right)\Psi \left( {{\Omega _1}\left( {{\varepsilon _i}} \right)} \right)} \right) + C} \right)},\\
{I_4^\infty} &= \frac{1}{{{R_1}\ln 2}}\sum\limits_{i = 1}^M \sum\limits_{n = 1}^{k - 1} \sum\limits_{j = 1}^n {\omega _i}\sqrt {1 - {\varepsilon _i}^2} {\Omega _1}\left( {{\varepsilon _i}} \right)\frac{{{\Psi ^n}\left( {{\Omega _1}\left( {{\varepsilon _i}} \right)} \right)}}{{n!}}\notag\\
&~~~\times {\left( {1 + \Psi \left( {{\Omega _1}\left( {{\varepsilon _i}} \right)} \right)} \right)} C_n^j{{\left( { - 1} \right)}^{n - j}}{{\left( {\Psi \left( {{\Omega _1}\left( {{\varepsilon _i}} \right)} \right)} \right)}^{ - j}} \notag\\
&~~~\times \Gamma \left( {j,\left( {\Xi  + 1} \right)\Psi \left( {{\Omega _1}\left( {{\varepsilon _i}} \right)} \right)} \right),
\end{align}
where $C \approx 0.57721$ is the Euler constant.
\end{corollary}
\begin{proof}
\emph{When $x \to 0$, the approximated expressions: $\exp \left( { - x} \right) \approx 1 - x$, ${\text{Ei}}\left( { - x} \right) \approx \ln \left( x \right) + C$ and $\gamma \left( {k,x} \right) \approx \frac{{x^k}}{k}$, can be utilized \cite{integration}.}
\end{proof}

\vspace{0.2cm}
\begin{corollary}
When $\rho \to \infty$, the approximated expression of the ergodic rate for the far user is expressed as
\vspace{-0.2cm}
\begin{align}
R_{far}^\infty  &= \frac{{{a_{far}}}}{{2{a_{near}}\ln 2}}\sum\limits_{i = 1}^M {\omega _i}\sqrt {1 - {\varepsilon _i}^2} \notag\\
&~~~ \times \left( {\left( {\frac{1}{{{\Omega _2}\left( {{\varepsilon _i}} \right) + 1}}} \right) - \delta\Phi^k\left( {{\Omega _2}\left( {{\varepsilon _i}} \right)} \right)} \right),
\end{align}
where $\delta = \frac{{2\left( {{{\left( {{R_2}^2 + {H^2}} \right)}^{\frac{{\alpha k}}{2} + 1}} - {{\left( {{R_1}^2 + {H^2}} \right)}^{\frac{{\alpha k}}{2} + 1}}} \right)}}{{k!\left( {R_2^2 - R_1^2} \right)\left( {\alpha k + 2} \right)}}$.
\end{corollary}
\vspace{0.1cm}
\begin{proof}
\emph{The result is obtained by substituting the asymptotic expression $\gamma \left( {k,x} \right) \approx \frac{{x^k}}{k}$ into (B.2) and applying some manipulations.}
\end{proof}

\vspace{0.2cm}
\subsubsection{High SNR slopes}
Furthermore, we evaluate the high SNR slope which is a key parameter determining the scaling law of the ergodic rate in high SNR region.
\begin{remark}\label{remark1}
Base on the approximated ergodic rates in the hign SNR region, the high SNR slopes of the near users and the far users are calculated as
\begin{align}
&S_{near} = \mathop {\lim }\limits_{\rho  \to \infty } \frac{{{R_{near}^\infty \left( \rho  \right) }}}{{\log \left( \rho  \right)}}=1,\\ 
&S_{far} = \mathop {\lim }\limits_{\rho  \to \infty } \frac{{{R_{far}^\infty\left( \rho  \right) }}}{{\log \left( \rho  \right)}}=0.
\end{align}
\end{remark}

\begin{remark}
As can be seen from (20) and (21), the high SNR slopes are fixed as constants, illustrating that the scaling rates of the ergodic rates in high SNR region are not affected by the number of STAR-RIS elements. 
\end{remark} 

\section{Numerical Results}
In this section, numerical simulation results are provided to verify the derived analytical results. The simulation is upon the original multiple i. i. d. Nakagami-$m$ fading channels. Without otherwise specification, the simulation parameter settings are defined as follows. A three-dimensional coordinate system with the center of the area as the origin and the units in meters is used to describe the location of each node: the BS is located at $(400, 0, 0)$ and the STAR-RIS is at $(0, 0, 30)$. The radii of the inner and outer circles are $R_1=100$ m and $R_2=200$ m. The noise power is ${\sigma ^2} = -170 + 10\log \left( {BW} \right) + {N_f} = -90$ dB, where the bandwidth $BW$ is $10$ MHz and the noise figure $N_f$ is 10 dB. The path loss exponent $\alpha$ equals to 2.6. The power allocation coefficients are $\alpha_{near}=0.3$ and $\alpha_{far}=0.7$. The SIC thresholds is $\gamma _{th}^{SIC} = 1$. The numbers of STAR-RIS elements are defined as $N=\left\{30, 50, 70\right\}$ and the parameters of the i.i.d. Nakagami-$m$ fading channels are $\left( {m,\Omega } \right) = \left( {2, 1} \right)$. Hence, base on (\ref{equ2}), the parameters of the fitting gamma distributions used in the analytical expressions are calculated as $\left( {k,\theta } \right) = \left\{ {\left( {3, 14} \right), \left( {5, 23.4} \right), \left( {7, 32.8} \right)} \right\}$.

We also perform the simulation for the spatially correlated channels to gain a better insight of the performance difference between the systems with the independent channels and the correlated channels. With other parameter settings same as above, each STAR-RIS element has a square shape with the side length ${\lambda  \mathord{\left/
 {\vphantom {\lambda  4}} \right.
 \kern-\nulldelimiterspace} 4}$, where $\lambda$ is the signal wavelength. The elements are deployed edge-to-edge in a 2 dimensional array: $5 \times 6$ for $N=30$, $5 \times 10$ for $N=50$ and $7 \times 10$ for $N=70$.

In order to compare the STAR-RIS with the conventional RIS, we simulate the ergodic rates of a RIS aided NOMA system for the same scenario. Since the conventional RIS can only reflect the signal, both the BS and the users have to be on the same side of the RIS. Thus the RIS is set at $(-200, 0, 50)$, the edge opposite the BS. The other settings are the same as the settings of the STAR-RIS.

Fig. \ref{Fig2} and Fig. \ref{Fig3} show the ergodic rates with the unit as bit per channel use (BPCU) versus the transmit SNR in dB with different $k$, which indicates that the analytical results match the simulation results well. Based on the two figures, it can be observed that 1) the ergodic rates positively correlate with the shape parameter of the fitted gamma distribution, that is, the number of STAR-RIS elements; 
\begin{figure}[!htb]
\centering
\includegraphics[width= 3in]{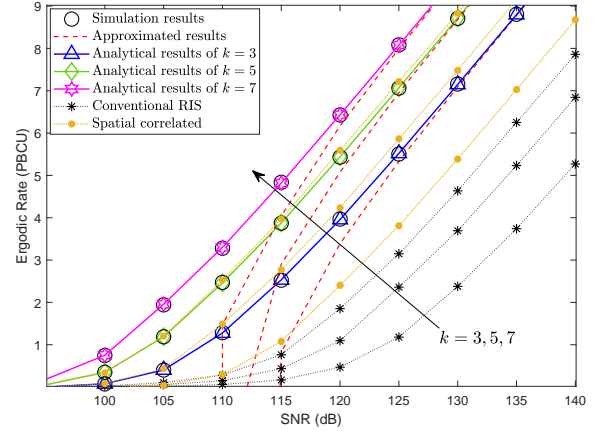}
\caption{Ergodic rates of near user versus transmit SNR, with $k=\left\{3, 5, 7\right\}$.}
\label{Fig2}
\end{figure}
\begin{figure}[!htb]
\centering
\includegraphics[width= 3.05in]{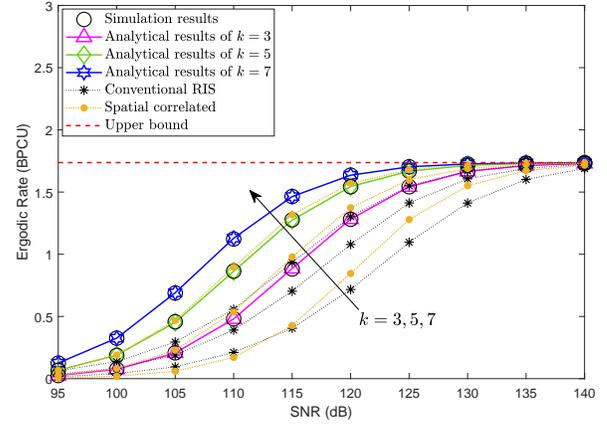}
\caption{Ergodic rates of far user versus transmit SNR, with $k=\left\{3, 5, 7\right\}$.}
\label{Fig3}
\end{figure}
2) in the high SNR region, the ergodic rates of the near user increase with the same rate for different $k$ and the far users converge to the same upper bound determined by the ratio $\frac{{{a_{far}}}}{{{a_{near}}}}$, which is consistent with \textbf{Remark \ref{remark1}} and \textbf{2}; 3) the ergodic rates of the system utilizing a STAR-RIS are obviously higher than utilizing a conventional RIS; 4) the ergodic rates for the uncorrelated channels outperform the correlated channels, but they have similar trends and their difference decreases as the number of the elements increases.

\section{Conclusions}
This letter has investigated the ergodic rates of STAR-RIS aided NOMA systems over Nakagami-$m$ fading channels for cell-edge users without LoS communication links from BSs. By fitting the distribution of the composite channel power gain to a gamma distribution, we have derived the closed-form expression of the ergodic rates as well as the high SNR slopes of the NOMA users. The numerical results indicate that: 1) the analytical results match the simulation results well and the approximated ergodic rates match the simulation results in the high SNR region; 2) the ergodic rates of the users can be improved by increasing the number of STAR-RIS elements; 3) the high SNR slopes of the users are fixed as constants, which are irrelevant to the number of STAR-RIS elements; 4) the performance of the STAR-RIS aided NOMA system is better than the RIS aided NOMA system.

\section*{Appendix~A: Proof of Theorem 1} \label{Appendix:A}
\renewcommand{\theequation}{A.\arabic{equation}}
\setcounter{equation}{0}
By applying some manipulates, (6) can be derived as
\begin{align}
    {R_{near}} &=\frac{U}{{\ln 2}}\int_0^\infty {\log \left( {1 + {\gamma _{near}}} \right){f_{{\gamma _{near}}}}\left( x \right)} dx \notag \\
    &= \frac{U}{{\ln 2}}\int_0^\infty  {\log \left( {1 + {\gamma _{near}}} \right)} d\left[ {{F_{\gamma _{near}}}\left( x \right) - 1} \right] \notag \\ 
    &\mathop= \limits^{\left( a \right)} \frac{U}{{\ln 2}}\left( {\log \left( {1 + {\gamma _{near}}} \right)\left[ {{F_{\gamma _{near}}}\left( \infty  \right) - 1} \right]} \right) \notag \\
    &~~~- \frac{U}{{\ln 2}}\int_0^\infty  {\frac{{{F_{{\gamma _{near}}}}\left( x \right) - 1}}{{1 + x}}} dx \notag \\ 
    &\mathop= \limits^{\left( b \right)} \frac{U}{{\ln 2}}\int_0^\infty  {\frac{{1 - {F_{{\gamma _{near}}}}\left( x \right)}}{{1 + x}}} dx,
\end{align}
where $U=u\left( {\gamma _{near}^{SIC} - \gamma _{th}^{SIC}} \right)$, $F_{\gamma _{near}}$ denotes the cumulative distribution function (CDF) of $\gamma _{near}$, $(a)$ is obtained by integration-by-parts method, and $(b)$ is obtained by ${F_{\gamma _{near}}}\left( \infty  \right)=1$.

Substituting (3) and (4) into (A.1), the ergodic rate of the near user is expressed as
\begin{align}
&{R_{near}} = \underbrace {\int_0^\Xi  {\int_0^{{R_1}} {\frac{{1 - {F_{{{\left| {{h_k}} \right|}^2}}}\left( {\Xi {\Psi}\left( y \right)} \right)}}{{1 + x}}{{f_{{d_{near}}}}\left( y \right)}dydx} } }_{{I_1}}\notag\\ 
&~~~+ \underbrace {\int_\Xi ^\infty  {\int_0^{{R_1}} {\frac{{1 - {F_{{{\left| {{h_k}} \right|}^2}}}\left( {{\Psi}\left( y \right)x} \right)}}{{1 + x}}{{f_{{d_{near}}}}\left( y \right)}dydx} } }_{{J_1}}.    
\end{align}
Since ${\left| {{h_k}} \right|}^2$ is a gamma distribution, $I_1$ can be calculated as
\begin{align}
{I_1} = \frac{{2\ln \left( {\Xi  + 1} \right)}}{{R_1^2\ln 2}}\int_0^{{R_1}} {\left( {1 - \frac{{\gamma \left( {k,\Xi {\Psi}\left( y \right)} \right)}}{{\Gamma \left( k \right)}}} \right)ydy}. 
\end{align}

As the representation of the lower incomplete gamma function $\gamma \left( {k,x} \right) = \Gamma \left( k \right)\left( {1 - \exp \left( { - x} \right)\sum\limits_{n = 1}^{k-1} {\frac{{{x^n}}}{{n!}}} } \right)$ is valid when $k$ is a positive integer [10, eq. (8.352.6)], $J_1$ can be derived as
\begin{align}
{J_1} \!&= \!\frac{2}{{R_1^2\ln 2}}\sum\limits_{n = 0}^{k - 1} {\frac{1}{{n!}}} \!\!\int_0^{{R_1}} \!\!y\Psi^n\left( y \right) \!\!\int_\Xi ^\infty  \!{\frac{{{x^n}\exp \left( { - {\Psi}\left( y \right)x} \right)}}{{1 + x}}} dxdy .
\end{align}

$J_1$ can be further calculated by dividing it into three parts: $J_1=I_2+I_3+I_4$, where $I_2$, $I_3$ and $I_4$ are denoted as
\begin{align}
{I_2} &= \frac{2}{{R_1^2\ln 2}}\int_0^{{R_1}} y\underbrace {\int_\Xi ^\infty  {\frac{{\exp \left( { - {\Psi}\left( y \right)x} \right)}}{{1 + x}}} dx} _{{J_2}}dy,\\
{I_3} &= \frac{2}{{R_1^2\ln 2}}\sum\limits_{n = 1}^{k - 1} {\frac{{{{\left( { - 1} \right)}^n}}}{{n!}}} \int_0^{{R_1}} y\Psi^n\left( y \right)\exp \left( {{\Psi}\left( y \right)} \right) \notag\\
&~~~~\times\underbrace{\int_{\Xi  + 1}^\infty  {\frac{{\exp \left( { - {\Psi}\left( y \right)x} \right)}}{x}} dx} _{{J_3}}dy, \\
{I_4} &= \frac{2}{{R_1^2\ln 2}}\sum\limits_{n = 1}^{k - 1} {\sum\limits_{j = 1}^n {\frac{{C_n^j{{\left( { - 1} \right)}^{n - j}}}}{{n!}}} } \int_0^{{R_1}} y\Psi^n\left( y \right)\exp \left( {{\Psi}\left( y \right)} \right) \notag\\
&~~~~\times\underbrace{\int_{\Xi  + 1}^\infty  {{x^{j-1}}\exp \left( { - {\Psi}\left( y \right)x} \right)} dx} _{{J_4}}dy.
\end{align}
With the aid of [10, eq. (3.352.2) and (3.351.2)], the integrals $J_2$, $J_3$ and $J_4$ can be calculated. 

The final solution can be then obtained by applying the Chebyshev-Gauss-quadrature to $I_1$, $I_2$, $I_3$ and $I_4$, which defined as $\int_{ - 1}^1 {\frac{{f\left( x \right)}}{{\sqrt {1 - {x^2}} }}} dx \approx \sum\limits_{i = 1}^M {{\omega _i}f\left( {{x_i}} \right)}$, where ${x_i} = \cos \left( {\frac{{2i - 1}}{{2M}}\pi } \right)$ and ${\omega _i} = \frac{\pi }{M}$.

\section*{Appendix~B: Proof of Theorem 2} \label{Appendix:B}
\renewcommand{\theequation}{B.\arabic{equation}}
\setcounter{equation}{0}
Similar to (A.1), (7) can be rewritten as
\begin{align}
  R_{far}=\frac{1}{{\ln 2}}\int_0^\infty  {\frac{{1 - {F_{{\gamma _{far}}}}\left( x \right)}}{{1 + x}}} dx.
\end{align}

Substituting (5) into (B.1), the ergodic rate of the far users can be derived as
\begin{align}
{R_{far}} \!&= \!\frac{1}{{\ln 2}}\int_0^{\frac{a_{far}}{a_{near}}} \!\!\!\int_{{R_1}}^{{R_2}} \frac{1 - \frac{{\gamma \left( {k,\frac{{\Phi \left( x \right)}}{{({y^2} + {H^2})}^{ - \frac{\alpha }{2}}}} \right)}}{\Gamma \left( k \right)}}{{1 + x}}{f_{d_{far}}}\left( y \right) dy dx \notag\\
&= \frac{2}{{\left( {R_2^2 - R_1^2} \right)\ln 2}}\sum\limits_{n = 0}^{k - 1} {\frac{1}{{n!}}} \int_0^{\frac{{{a_{far}}}}{{{a_{near}}}}} \frac{{{\Phi ^n}\left( x \right)}}{{1 + x}}\notag\\
&\!\!\!\!\!\!\!\!\!\!\times \underbrace {{{\int_{{R_1}}^{{R_2}} {\left( {{{({y^2} + {H^2})}^{\frac{\alpha }{2}}}} \right)} }^n}\exp \left( { - \Phi \left( x \right){{({y^2} + {H^2})}^{\frac{\alpha }{2}}}} \right) ydy} _{{J_5}}dx.
\end{align}

With the help of [10, eq. (3.351.1)], $J_5$ can be calculated. We then employ the Chebyshev-Gauss-quadrature to obtain the final result.

\bibliographystyle{IEEEtran}
\bibliography{mybib}

\end{document}